\documentclass{article}
\usepackage{graphicx}
\usepackage{amsmath}
\usepackage{amsthm}
\usepackage{amssymb}
\usepackage{url}
\newcommand{\be}{\begin{equation}}
\newcommand{\ee}{\end{equation}}
\newcommand{\BC}{\mathbb{C}}
\newcommand{\BR}{\mathbb{R}}

\newcommand{\BP}{\mathbb{P}}
\newcommand{\bx}{\mathbf{x}}
\newcommand{\bn}{\mathbf{n}}
\newcommand{\diag}{\mathrm{diag}}
\newtheorem{Lemma}{Lemma}
\newtheorem{proposition}{Proposition}
\newtheorem{remark}{Remark}

\newtheorem{corollary}{Corollary}
\begin{document}
\title{GEOMETRY AND SHAPE OF MINKOWSKI'S SPACE CONFORMAL INFINITY}
\author{Arkadiusz Jadczyk\\
Center CAIROS, Institut de Math\'{e}matiques de
    Toulouse\\Universit\'{e} Paul Sabatier, 31062 TOULOUSE CEDEX� 9, France\\
    email: arkadiusz.jadczyk@cict.fr}
\maketitle
\begin{abstract}
We review and further analyze Penrose's 'light cone at infinity' - the conformal closure of Minkowski space. Examples of a potential confusion in the existing literature about it's geometry and shape are pointed out. It is argued that it is better to think about conformal infinity as of a needle horn supercyclide (or a limit horn torus) made of a family of circles, all intersecting at one and only one point, rather than that of a 'cone'. A parametrization using circular null geodesics is given. Compactified Minkowski space is represented in three ways: as a group manifold of the unitary group U(2), a projective quadric in six-dimensional real space of signature (4,2), and as the Grassmannian of maximal totally isotropic subspaces in complex four--dimensional twistor space. Explicit relations between these representations are  given, using a concrete representation of antilinear action of the conformal Clifford algebra Cl(4,2) on twistors. Concepts of space-time geometry are explicitly linked to those of Lie sphere geometry. In particular conformal infinity is faithfully represented by planes in 3D real space plus the infinity point. Closed null geodesics trapped at infinity are represented by parallel plane fronts (plus infinity point). A version of the projective quadric in six-dimensional space where the quotient is taken by positive reals is shown to lead to a symmetric Dupin's type `needle horn cyclide' shape of conformal infinity.
\end{abstract}
\section{Introduction\label{sec:intro}}
A persistent confusion about Minkowski's space conformal infinity started with a widely quoted paper by Roger Penrose `\textit{The light cone at infinity\,} \cite{penrose1}.  In the abstract to this seminal paper Penrose wrote:
\begin{quotation}
From the point of view of the conformal structure of space-time, ``points at infinity" can be treated on the same basis as finite points. Minkowski space can be completed to a highly symmetrical conformal manifold by the addition of a null cone at infinity - the ``absolute cone".
\end{quotation}
He then elaborated in the main text:
\begin{quotation}
Let $x^\mu$ be the position vector of a general event in Minkowski space-time relative to a given origin $O.$ Then the transformation to new Minkowskian coordinates $\hat{x}^\mu$ given by
\be
\hat{x}^\mu=\frac{x^\mu}{x_\alpha x^\alpha},\quad x^\mu=\frac{\hat{x}^\mu}{\hat{x}_\alpha \hat{x}^\alpha},\ee
is conformal (``inversion with respect to $O$"). Observe that the whole null cone of $O$ is transformed to infinity in the $\hat{x}^\mu$ system and that infinity in the $x^\mu$
system becomes the null cone of the origin $\hat{O}$ of the $\hat{x}^\mu$ system. (``Space--like" or ``time-like" infinity become $\hat{O}$ itself but ``null" infinity becomes spread out over the null cone of $O.$) Thus, from the conformal point of view ``infinity" must be a null cone.
\end{quotation}
Penrose's statement, ``\textit{that infinity in the $x^\mu$
system becomes the null cone of the origin $\hat{O}$ of the $\hat{x}^\mu$ system}" apparently had a confusing effect even on some experts in the field. For instance, in the monograph \cite[p. 127]{akivis}, we find the statement that ``\textit{`conformal infinity' is the result of the conformal inversion of the light cone at the origin of $M,$}" and in another monograph  Huggett and Tod write about the compactified Minkowski space $M^c$ \cite[p. 36]{tod}: ``\textit{Thus $M^c$ consists of $M$ with an extra null cone added at infinity.\,}" Not only they write so in words, but they also miss a part of the conformal infinity (the closing two--sphere) in their, otherwise excellent and clear, formal analysis.

This apparent  confusion has been described in \cite{ajci}, where also a deeper analysis of the structure of the conformal infinity has been given using, in particular, Clifford algebra techniques.
In \cite{aj9a} a close similarity has been noticed between the geometry and shape of the conformal infinity with that of a Dupin's type (super)cyclide. In the present paper we review and develop these ideas further on, and also make a step in relating them to Lie sphere geometry in $\BR^3$ developed by Sophus Lie \cite{lie1}, Wilhelm Blaschke \cite{blaschke3} and Thomas E. Cecil \cite{cecil1}.

In section \ref{sec:intro} we introduce the compactified Minkowski space $M^c$ (via Cayley's transform) following Armin Uhlmann \cite{uhlmann63}, as the group manifold of the unitary group $U(2),$ and the conformal infinity as the subset of $U(2)$ consisting of those matrices $U\in U(2)$ for which $\det (U-I)=0.$ In section \ref{sec:su22} we review the relation of the compactified Minkowski space and it's conformal infinity part to the group $SU(2,2)$ (the spin group of the conformal group), and to it's action on $U(2)$ via fractional linear transformations $U'=(AU+B)(CU+D)^{-1}.$ In particular the role of totally isotropic subspaces of $\BC^{2,2}$ (as null geodesics and as points of $M^c$) is elucidated there. In section \ref{sec:o42} the $SU(2,2)$ formalism is explicitly related to the $O(4,2)$ representation via a particular matrix realization (by antilinear transformations) of the Clifford algebra $Cl_{4,2}.$ The main results of this section are contained in Proposition \ref{pro:x} and Corollary \ref{cor:cor1}, where an explicit formula for a bijective map between the projective quadric of $\BR^{4,2}$ and $U(2)$ is given - cf. Eq. (\ref{eq:u}). Our conventions are: coordinates $x^\mu,\, \mu=1,..,4,$ with $x^4=ct,$ for the Minkowski space, $x^\alpha,\, \alpha=1,...,6$ for $\BR^{4,2}$ endowed with the quadratic form $Q(x)=(x^1)^2+(x^2)^2+(x^3)^2-(x^4)^2+(x^5)^2-(x^6)^2.$ Let $Q=\{x\in \BR^{4,2}:\, Q(x)=0,\, x \neq 0\},$ We discuss two equivalence relations in $\BR^{4,2}:$ the standard one in projective geometry, $x\sim y$ iff $x=\lambda y,\, \lambda\in \BR^*=\BR\setminus\{0\},$ and a stronger one $x\approx y$ iff $x=\lambda y,\, \lambda >0.$ In section \ref{sec:dci} we discuss $\tilde{M}^c,$ the double covering of $M^c,$ defined as $Q/\approx, $ and the corresponding conformal infinity. Skipping one space dimension, and projecting from four dimensions on a 3D box, the conformal infinity has the shape of an elliptic supercyclide as depicted in Fig. 1. Simple conformal infinity, that of $M^c,$  is discussed in section \ref{sec:sci} where we represent it in two ways: as an asymmetric  needle cyclide in Fig. 2, and as a symmetric limit torus in Fig. 3. In section \ref{sec:ls}, in particular cf. Table 1 adapted from \cite{cecil1}, the correspondence between the objects of the space of Lie spheres and those of $\BR^{4,2}$ geometry is described, and then used for elucidating the $\BR^3$ picture of conformal infinity. A null geodesic trapped at infinity can be represented as a family of plane fronts in $\BR^3$ - cf. Fig. 4, or, equivalently, as a path on the supercyclide intersecting its cusp - Fig. 5. The family of such null geodesics essentially determines the geometry of the conformal infinity which carries a natural conformal structure of signature $(2,0).$
\section{Minkowski's space conformal infinity}
Albert Einstein introduced the Minkowski space as the `affine space of events' equipped with the Minkowskian infinitesimal line element $ds^2=(dx^1)^2+(dx^2)^2+(dx^3)^2-(dx^4)^2$, and this is the most popular image today.\footnote{We use $x^4=ct$ rather than more popular $x^0.$} `Affine' means that there is no distinguished `origin', though each inertial observer selects one particular event as having all four coordinates zero in the coordinate system of his frame of reference. Mathematically equivalent is another approach: Minkowski space is a four-- dimensional real vector space, endowed with the quadratic form $q(x)=(x^1)^2+(x^2)^2+(x^3)^2-(x^4)^2,$ but when studying its geometry we are looking for geometrical objects, concepts and constructions that are invariant under the full 10-parameter Poincar\'{e} group consisting of Lorentz transformations and translations.
Poincar\'{e}'s group is the fundamental symmetry group of all relativistic theories. But, in fact, this very group appeared naturally in the works of geometers of the XIX-th century studying the `space of (Lie) spheres' in $\BR^3$, cf. \cite{lie1,cecil1}, in a way that had nothing to do with the philosophy of relativity.

Let us introduce the notation that will be used in the following. Minkowski space will be denoted, alternatively, either as $M,$ or as $E^{3,1},$ or as $\BR^{3,1}.$ We will represent it as a vector space endowed with the scalar product $(x,y)=x^1y^1+x^2y^2+x^3y^3-x^4y^4.$ Introducing the metric tensor
$\eta=\mbox{diag}(1,1,1,-1),$ the scalar product is written as $(x,y)=\eta_{\mu\nu}\,x^\mu y^\nu=\eta^{\mu\nu}x_\mu y_\nu.$ The Lorentz group $L=O(3,1)$ is the group of all $4\times4$ real matrices $\Lambda$ for which ${}^t\Lambda\,\eta\,\Lambda=\eta.$ It acts on $M$ via linear transformations $x^\mu\mapsto {\Lambda^\mu}_\nu\,x^\nu.$ Translation group $T,$ isomorphic to the additive group of $\BR^4,$ acts on $M$ via $x^\mu\mapsto x^\mu+a^\mu.$ The Poincar\'{e} group $P$ is the semidirect product of $L$ and $T.$ It consists of pairs $(a,\Lambda),$ and acts on $M$ via $x^\mu\mapsto {\Lambda^\mu}_\nu x^\nu+a^\mu.$ That implies the composition law of the semidirect product: $(a,\Lambda)(a',\Lambda')=(a+\Lambda a', \Lambda\Lambda').$

In quantum theory we are interested in ray representations of the Po\-in\-ca\-r\'{e} group on complex vector spaces. Ray representations lead to vector representations of the double covering group. This way we are led from the Lorentz group to its double covering group - $SL(2,\BC),$ the group of unimodular (i.e. of determinant one) complex  $2\times 2$ matrices. It's action on $M$ is then conveniently coded via standard Hermitian Pauli's matrices $\sigma_\mu,$ where we put $\sigma_4=\left(\begin{smallmatrix}1&0\\0&1\end{smallmatrix}\right)=I.$  The mapping $x\mapsto \sigma(x)=x^\mu\sigma_\mu$ maps bijectively $M$ onto the space of $2\times 2$ Hermitian matrices, with the important property that $q(x)=\det(\sigma(x)). $ If $A\in SL(2,\BC),$ then $A\sigma(x)A^\dagger$ is Hermitian, thus $A\sigma(x)A^\dagger=\sigma(x'),$ and since $\det(\sigma(x'))=\det(\sigma (x)),$ we have $q(x)=q(x').$ It follows that $x'$ is related to $x$ by a Lorentz transformation: $A\sigma(x)A^\dagger=\sigma(\Lambda(A)x).$ The mapping $SL(2,\BC)\ni A\mapsto \Lambda(A)\in L$ is then a group homomorphism from $SL(2,\BC)$ onto the connected component of identity of $L,$ with kernel $\{I,-I\}.$

There are two simple ways in which Hermitian matrices can be transformed into unitary matrices. The first one is by exponentiation: $X\mapsto\exp(iX).$ It is not very interesting here, as it is periodic. The second way, more interesting in the present context, is by Cayley's transform $X\mapsto u(X)=U=\frac{X-iI}{X+iI}.$ The inverse transform $u^{-1}(U)=X=i\,\frac{I+U}{I-U}$ is well defined whenever $\det(I-U)\neq 0.$ The space $U(2)$ of $2\times 2$ (complex) unitary matrices is a four--dimensional (real) compact manifold, and Cayley's transform maps $M$ onto an open dense submanifold of $U(2).$ The remaining part, described by the algebraic equation $\det(U-I)=0$ is what is being called \textit{the conformal infinity of $M$} \cite{uhlmann63}.
\section{The group $SU(2,2)$ \label{sec:su22}}
Early in the XX-th century (1909--1910) Bateman and Cunningham \cite{bateman1,cunningham1,bateman2} established local invariance of the wave equation and of Maxwell's equations under conformal transformations. The central role in these transformations is being played by the conformal inversion $R$, formally defined by \be R:\, (\bx,t)\mapsto r_0^2\frac{(\bx,t)}{\bx^2-c^2t^2},\label{eq:ci}\ee where $r_0$ is a constant of physical dimension of length. Conformal inversion is singular on the light cone $q(x)=\bx^2-c^2t^2=0.$ Together with Poincar\'{e} group transformations, it generates the conformal group of local transformations of $M,$ isomorphic to $O(4,2).$
The spin group for the conformal group, in our settings the group $SU(2,2),$ enters the scene through the following observations.

Let $ G $ be the matrix $ G =\diag(1,1,-1,-1).$ Then $U(2,2)$ is the group of $4\times 4$ complex matrices $\mathcal{U}$ with the property $\mathcal{U} G \mathcal{U}^\dagger= G ,$ where ${}^\dagger$ denotes the Hermitian conjugation. Writing $\mathcal{U}$ in the $2\times 2$ block matrix form as $\mathcal{U}=\left(\begin{smallmatrix}A&B\\C&D\end{smallmatrix}\right),$ the condition $\mathcal{U} G \mathcal{U}^\dagger= G ,$ translates into $A^\dagger A-C^\dagger C=D^\dagger D-B^\dagger B=I$ and $A^\dagger B-C^\dagger D=0.$ The group $SU(2,2)$ acts naturally on $U(2)$ by fractional linear transformations:
\be U\mapsto U'=(AU+B)(CU+D)^{-1}.\label{eq:action}\ee
Namely, with some little effort, one can show that if $U$ is unitary, then $CU+D$ is invertible and that $(AU+B)(CU+D)^{-1}$ is again unitary. Evidently the matrix $(AU+B)(CU+D)^{-1},$ is insensitive to the overall complex phase of $\mathcal{U},$ therefore, effectively, we can restrict ourself to the subgroup $SU(2,2)$ by requiring $\det(\mathcal{U})=1.$ This way the compactified Minkowski space, which we will denote as  $M^c,$  the group manifold of $U(2),$ becomes a homogeneous space for the group $SU(2,2).$\footnote{In fact, $M^c$ is the Shilov's boundary of the bounded homogeneous complex domain $SU(2,2)/S(U(2)\times U(2)),$ cf. e.g. \cite{cj}, but we will not need this fact and its consequences here.}


Now, having the group $U(2),$ with its distinguished group identity element $U_0=I,$ as a homogeneous space does not look very natural. Therefore, taking the group $SU(2,2)$ (or a group isomorphic to it) as a basic element, a more abstract and more `basic'  construction is needed. To this end one may  choose a coordinate free construction, starting from what is often called `the twistor space'\footnote{For a clear, concise and mathematically precise introduction see e.g. \cite{wk}, also references therein.} Thus let $V$ be a complex vector space equipped with a pseudo-Hermitian form, written as $\langle v|w\rangle,$ of signature $(2,2).$ A basis $e_i$ in $V$ is called orthonormal if $\langle e_i|e_j\rangle = G _{ij},\, (i,j=1,..,4).$ Any two orthonormal bases $e',e$ are then related by a $U(2,2)$ transformation $e'_i=e_j\,{\mathcal{U}^j}_i.$ In order to be able to reduce the transformation group to $SU(2,2)$ a volume form $\omega$ is selected in $\bigwedge^4 V,$ and the set of orthonormal bases is reduced to those having the property $e_1\wedge ...\wedge e_4=\omega.$ The relation to space--time geometry is now obtained via the study of one- and two--dimensional totally isotropic subspaces of $V.$\footnote{One could think that the term \textit{isotropic} subspace should be enough, since if a subspace has all its vectors isotropic, then any two its vectors must be, automatically, orthogonal. However, in the literature, by an isotropic subspace one usually means `a subspace that contains a non-zero isotropic vector.' Therefore, in order to avoid the confusion, the additional adjective `totally' is needed for a subspace whose any two vectors are mutually orthogonal.} Two--dimensional totally isotropic subspaces of $V$ correspond to points in the compactified Minkowski space $M^c,$  while one--dimensional isotropic subspaces of $V$ correspond to `null geodesics' in $M^c$ \cite{wk}. This correspondence has a remarkable geometric simplicity and beauty: if $v$ is an isotropic vector representing a null geodesic in $M^c$, then the set of all totally isotropic subspaces containing $v$ is the set of points in $M^c$ on this geodesic. If $W$ is a two--dimensional totally isotropic subspace representing a point $p$ in $M^c,$ then nonzero vectors (automatically isotropic) of $W$ are null geodesics through $p$. If two isotropic planes intersect - then the corresponding points in $M^c$ can be connected by a null geodesic. If two isotropic vectors in $V$ are mutually orthogonal, the corresponding geodesics intersect.
\subsection{Relation between $U(2)$ and $SU(2,2)$ pictures}
In this subsection we will describe the relation between the two pictures of $M^c, $ one as the set of all $2\times 2$ unitary matrices, and one as the set of totally isotropic planes in $V.$  To this end we choose an orthonormal basis $e_i$ in $V$ and split $V$ into a direct sum $V=\BC^2\oplus\BC^2.$ Thus each vector in $V$ can be written as column $\left(\begin{smallmatrix}u\\v\end{smallmatrix}\right)$ with $u$ being a linear combination of $e_1,e_2,$ and  $v$ of $e_3,e_4.$ It is then easy to see that each totally isotropic subspace of $V$ is uniquely represented in the form $\left(\begin{smallmatrix}Uv\\v\end{smallmatrix}\right),$ where $v$ runs through $\BC^2$ spanned by $e_3,e_4,$ and $U$ is a unitary operator in this space. Moreover, if $\left(\begin{smallmatrix}A&B\\C&D\end{smallmatrix}\right),$ is in $U(2,2),$ then $$\left(\begin{smallmatrix}A&B\\C&D\end{smallmatrix}\right)\left(\begin{smallmatrix}Uv\\v\end{smallmatrix}\right)=\left(\begin{smallmatrix}(AU+B)v\\(CU+D)v\end{smallmatrix}\right)=\left(\begin{smallmatrix}U'v'\\v'\end{smallmatrix}\right),$$
where $U'=(AU+B)(CU+D)^{-1},$ and $v'=(CU+D)v.$ Since, as we mentioned before, $CU+D$ is necessarily invertible, $v'$ runs through the whole $\BC^2$ when $v$ does so.
\section{$\BR^{4,2}$ and the group $O(4,2)$ \label{sec:o42}}
Let $\BR^{4,2}$ be $\BR^6$ endowed with the quadratic form $Q(x)=(x^1)^2+(x^2)^2+(x^3)^2-(x^4)^2+(x^5)^2-(x^6)^2$ and the associated pseudo-Hermitian form $\langle x,y\rangle =x^1y^1+x^2y^2+x^3y^3-x^4y^4+x^5y^5-x^6y^6.$
We start with the following Proposition essentially taken from \cite{aj9a}, and refer the Reader there for more details, though, in fact, the proof is nothing but a somewhat tedious, simple calculation.\footnote{The author does not know whether these properties are known to the experts or not. Any hint to the existing literature will be appreciated.}
\begin{proposition}
Consider the following set of six complex $4\times 4$ matrices:
$$ \Gamma_1=\left(\begin{smallmatrix}
 0 & 0 & i & 0 \\
 0 & 0 & 0 & -i \\
 i & 0 & 0 & 0 \\
 0 & -i & 0 & 0
\end{smallmatrix}
\right)\, \Gamma_2=\left(
\begin{smallmatrix}
 0 & 0 & 1 & 0 \\
 0 & 0 & 0 & 1 \\
 1 & 0 & 0 & 0 \\
 0 & 1 & 0 & 0
\end{smallmatrix}
\right)\,
\Gamma_3=\left(
\begin{smallmatrix}
 0 & 0 & 0 & -i \\
 0 & 0 & -i & 0 \\
 0 & -i & 0 & 0 \\
 -i & 0 & 0 & 0
\end{smallmatrix}
\right)$$
$$\Gamma_4=\left(
\begin{smallmatrix}
 0 & i & 0 & 0 \\
 -i & 0 & 0 & 0 \\
 0 & 0 & 0 & i \\
 0 & 0 & -i & 0
\end{smallmatrix}
\right)\quad \Gamma_5=\left(
\begin{smallmatrix}
 0 & 0 & 0 & -1 \\
 0 & 0 & 1 & 0 \\
 0 & 1 & 0 & 0 \\
 -1 & 0 & 0 & 0
\end{smallmatrix}
\right)\quad
\Gamma_6=\left(
\begin{smallmatrix}
 0 & 1 & 0 & 0 \\
 -1 & 0 & 0 & 0 \\
 0 & 0 & 0 & -1 \\
 0 & 0 & 1 & 0
\end{smallmatrix}
\right).$$
For each $x=(x^1,...,x^6)\in \BR^{4,2},$ let $X$ be the matrix
\be X=\sum_{\alpha=1}^6 x^\alpha \Gamma_\alpha,\label{eq:X}\ee
Then a straightforward calculation shows that these matrices satisfy the following relations:
\begin{enumerate}
\item[{\rm (i)}] $GXG = -\,{}^tX,$
\item[{\rm (ii)}] $\bar{{X}^i}_j=\frac{1}{2}\epsilon^{imnk}\,G_{mj}G_{nl}\;{X^l}_k,$
\item[{\rm (iii)}] $X\bar{Y}+Y\bar{X}=2\langle x,y\rangle,$
\item[{\rm (iv)}] $\det(X)=Q(x)^2,$
\item[{\rm (v)}]If $R\in SU(2,2),$ then $R\Gamma_\alpha R^{-1}=\Gamma_\beta {L(R)^\beta}_\alpha,$ and $R\mapsto L(R)$ is a group homomorphism from $SU(2,2)$ onto the connected component of identity $SO_+(4,2),$ with kernel $\{1,-1,i,-i\}.$
\item[{\rm (vi)}] The 15 matrices $L_{\alpha\beta}=\Gamma_\alpha\bar{\Gamma}_\beta-\Gamma_\beta\bar{\Gamma}_\alpha,$ $\alpha<\beta,$ form a basis of the Lie algebra of $SU(2,2).$
\end{enumerate}
\label{pro:x}\end{proposition}
\begin{remark}
The meaning of {\normalfont (iii)} is that the mapping $x\mapsto \hat{X},$ where $\hat{X}$ is the antilinear operator on $C^4$ defined by $({\hat{X}v})^i={X^i}_j\,\bar{v}^j,$ is a Clifford map from $\BR^{4,2}$ to the algebra of all real-linear transformations of $\BC^4.$ The algebra $\mbox{Mat}(4,\BC),$ as an algebra over $\BR,$ can be then identified with the even Clifford subalgebra of $\BR^{4,2}.$
\end{remark}
\subsection{Compactified Minkowski space $M^c$ as a projective quadric in $\BR^{4,2}$}
Probably the most popular representation of $M^c$ that can be found in the literature is one where $M^c$ is defined as the set of generator lines of the cone (minus the origin $\{0\}$) \footnote{For a more general discussion of the case of signature $(r,s)$ see, for instance, \cite[Ch. 1.4.3]{angles}} $$\mathcal{C}=\{x\in\BR^{4,2}:x\neq 0,\quad Q(x)=0\}.$$ Or, in other words, it is the manifold of all one--dimensional isotropic subspaces  of $\BR^{4,2}.$ Or else, it is the cone $\mathcal{C}$ divided by the equivalence relation: $x\sim y$ if and only if $x=\lambda y,$ $\lambda\neq 0,\, \lambda\in \BR .$ We denote the resulting projective quadric, consisting of equivalence classes $[x]$ of non-zero isotropic vectors $x\in\BR^{4,2},$ by $[\mathcal{C}]=\mathcal{C}/\sim .$ It is now important to know the explicit relation between $M^c$ defined as $[\mathcal{C}]$ and $U(2).$ This is given by the following corollary to our Proposition \ref{pro:x}:
\begin{corollary}
For each $x\in \mathcal{C}$ the matrix \be U(x)=\frac{1}{x^4+ix^6}\begin{pmatrix}-x^3+ix^5&-x^1+ix^2\\-x^1-ix^2&x^3+ix^5\end{pmatrix}\label{eq:u}\ee
is unitary and depends only on the equivalence class $[x]$ of $x.$ We have
\be \det (U(x)-I)=-\frac{2i(x^5-x^6)}{x^4+ix^6}.\ee
Therefore $\det (U(x)-I)=0$ if and only if $x^5=x^6.$\label{cor:cor1}
\end{corollary}
While the proof of this corollary is by a straightforward calculation, the deeper meaning of it is revealed by a study of the kernels of Clifford algebra representatives on a Clifford module, as discussed in \cite[Eq. (20)]{aj9a}. $\BC^4,$ when considered as $\BR^8,$ is a module for the Clifford algebra $Cl_{4,2},$ the map $x\mapsto \hat{X}$ being the Clifford map. One then computes the kernel of $\hat{X},$ which is then represented in the form $\left(\begin{smallmatrix}Uv\\v\end{smallmatrix}\right).$ By defining $U(x)=U$ one gets the formula (\ref{eq:u}).

\textbf{The compactified Minkowski space is this way represented as a projective quadric described by the equation $$Q(x)=(x^1)^2+(x^2)^2+(x^3)^2-(x^4)^2+(x^5)^2-(x^6)^2=0$$ in $\BR\BP^5.$ The conformal infinity is an intersection of this projective quadric with the projective hyperplane $$x^5=x^6.$$}
\section{The doubled conformal infinity as an elliptic supercyclide\label{sec:dci}}
The conformal infinity is a real algebraic variety described in homogeneous coordinates by two homogeneous equations: $Q(x)=(x^1)^2+(x^2)^2+(x^3)^2-(x^4)^2+(x^5)^2-(x^6)^2=0$ (compactified Minkowski space) and $x^5=x^6$ (the infinity hyperplane). In this section we will replace the equivalence relation in $\BR^6-\{0\}:$ $x\sim y$ iff $x=ry,\, r \neq 0,$ by a stronger one $x\approx y$ iff $x=ry,\, r>0.$ The doubled compactified Minkowski space $\tilde{M}^c$ is defined as the quotient of $\{x: Q(x)=0\}/\approx .$ \footnote{Topologically $M^c$ and $\tilde{M}^c$ are equivalent. Indeed $M^c$ is topologically $U(2)$ which is $(U(1)\times SU(2))/\{I,-I\}.$ $\tilde{M}^c$ is topologically $U(1)\times SU(2),$ (no quotient). But both spaces are homeomorphic, since $U(2)$ can be parametrized also as $S^1\times S^3:$ $U=\left(\begin{smallmatrix}z_1&-c\bar{z_2}\\z_2&c\bar{z}_2\end{smallmatrix}\right),$ $ |c|=1,\,|z_1|^2+|z_2|^2=1.$} We can embed now $M=\BR^{3,1},$ described by coordinates $(\bx,t)$ in $\tilde{M}^c$ in two ways:
$$\phi_+(\bx,t)=[(\bx,t,\frac{1}{2}(1-\bx^2+t^2),-\frac{1}{2}(1+\bx^2-t^2))],$$
$$\phi_-(\bx,t)=[(\bx,t,-\frac{1}{2}(1-\bx^2+t^2),\frac{1}{2}(1+\bx^2-t^2))].$$
The first embedding is characterized by the equation $x^5-x^6=1,$ the second one by  $x^5-x^6=-1.$ As we will see, in $\tilde{M}^c$ there are also two special, singular points: $[(\mathbf{0},0,1,1)]$ and  $[(\mathbf{0},0,-1,-1)].$
\subsection{Graphic representation as a needle horn \label{sec:gr}}
To obtain a geometric representation of the conformal infinity in $\tilde{M}^c$ consider the two defining equations written as
\be (x^1)^2+(x^2)^2+(x^3)^2+(x^5)^2=(x^4)^2+(x^6)^2,\label{eq:ss}\ee
\be x^5=x^6.\ee
Clearly the number $(x^1)^2+(x^2)^2+(x^3)^2+(x^5)^2=(x^4)^2+(x^6)^2$ is positive, it cannot be zero because that would imply $x=0, $ and the origin is excluded. Therefore we can always choose a unique positive scaling factor and get two equations in $\BR^6:$ $(x^1)^2+(x^2)^2+(x^3)^2+(x^5)^2=1,$ and $(x^4)^2+(x^6)^2=1.$ These are two intersecting cylinders. The infinity plane $x^5=x^6$ cuts this intersection effectively reducing the number of dimensions to 3. We obtain:
 \be (x^1)^2+(x^2)^2+(x^3)^2+(x^5)^2=1,\ee
 \be (x^4)^2+(x^5)^2=1.\ee
 In order to arrive at a graphics representation in $\BR^3$ we suppress one space dimension, say $x^3,$ so that two-spheres will be represented by circles. We are left now with four variables  $(x^1,x^2,x^4,x^5), $ and the intersection of two cylinders:
\be (x^1)^2+(x^2)^2+(x^5)^2=1,\label{eq:s1}\ee
\be (x^4)^2+(x^5)^2=1,\label{eq:s2}\ee
in $\BR^4.$ We now choose a light source in $\BR^4,$ a $3D$ box, and project our surface onto the box. For the light source we choose the point $x_0$ with coordinates $x^1=2,x^2=x^4=x^5=0$ (it can be easily verified that the whole represented body is contained inside a sphere of radius $1$), for the screen let us choose the space $(0,x^2,x^4,x^5).$ The screen will cut our surface, but this is not a problem. From now on let us call the screen variables $(x,y,z).$ The straight line in $\BR^4$ connecting the source $(2,0,0,0)$ with a point $(x^1,x^2,x^4,x^5)$ has the parametric equation: $$x(s)=(1-s)(2,0,0,0)+s(x^1,x^2,x^4,x^5)=(s(x^1-2)+2,sx^2,sx^4,sx^5).$$
It cuts the screen when $s(x^1-2)+2=0,$ therefore for $s=2/(2-x^1).$ This way it hits the screen at the point $(0,sx^2,sx^4,sx^5),$ which gives us the equations for the image:
\begin{align}
x(x^2,x^4,x^5)&=\frac{2x^2}{2-x^1}\\
y(x^2,x^4,x^5)&=\frac{2x^4}{2-x^1}\\
z(x^2,x^4,x^5)&=\frac{2x^5}{2-x^1}.
\end{align}
Let us choose now angular coordinates for the variables $x^1,x^2,x^4,x^5$ satisfying Eqs (\ref{eq:s1},\ref{eq:s2}). To satisfy (\ref{eq:s2}) we set
\begin{align} x^4&=\sin\,\Theta,\\x^5&=\cos\,\Theta.
\end{align}
Then, from (\ref{eq:s1}), we get $(x^1)+(x^2)^2=1-(x^4)^2=\cos^2\,\Theta,$ and as long $\cos\, \Theta \neq 0,$ (the two singular points), we can set, uniquely,
\begin{align} x^1&=\cos\,\Psi\cos\,\Theta,\\x^2&=\sin\,\Psi\cos\,\Theta .\end{align}
After substitution of these parametrization into the surface equation we get
\begin{align}
x(\Psi,\Theta)&=\frac{2\sin\,\Psi\cos\,\Theta}{2-\cos\,\Psi\cos\,\Theta}\\
y(\Psi,\Theta)&=\frac{2\sin\,\Theta}{2-\cos\,\Psi\cos\,\Theta}\\
z(\Psi,\Theta)&=\frac{2\cos\,\Theta}{2-\cos\,\Psi\cos\,\Theta}.
\end{align}
These are the equations of a degenerate {\it elliptic supercyclide\,} \cite[Eq. (11)]{garnier}, which is a slightly deformed Dupin's cyclide known under the names \textit{needle (horn) cyclide\,} \cite[Fig. 8, p. 83]{schrott}, \cite[Fig. 5.11, p. 158]{cecil1}, or, in French, \textit{double croissant sym\'{e}trique\,} \cite{ferreol}.
\begin{figure}[!ht]
\begin{center}
      \includegraphics[width=7cm, keepaspectratio=true]{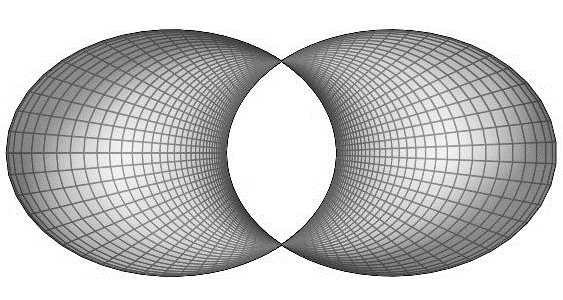}
\end{center}
  \caption{Pictorial representation of the doubled conformal infinity  with one dimension skipped (elliptic supercyclide).}
\label{fig:fig1}\end{figure}
The  simplest  form  of  the  cyclide  may  be  thought  of  as  a  deformed  torus,  in which  the  minor radius  varies  around  the  central  hole. In particular the  Dupin  cyclides  provide  a  generalization  of  all  the  surfaces  conventionally  used  in  solid  modeling - the  plane,  cylinder,  cone,  sphere  and  torus \cite{pratt}.
\section{Simple conformal infinity\label{sec:sci}}
By taking the quotient, as in section (\ref{sec:dci}), but by $\BR^*=\BR\setminus \{0\}$ rather than by $\BR^+,$ we arrive at the same equations (\ref{eq:s1},\ref{eq:s2}), but this time $x$ and $-x$ describe the same point.

Jakob Steiner has faced a similar problem when studying the method of representing the projective plane in $\BR^3.$ One possible solution was to use quadratic expressions in the coordinates - cf \cite{artmann} and \cite[p. 340]{hilbert}. Let us first follow a similar method. In order to represent the resulting variety graphically, we will need the following lemma:
\begin{Lemma}
With the notation as in section \ref{sec:gr} introduce the following variables:
\be
y_\alpha = x^\alpha x^4.
\ee
Then, assuming that $x^\alpha,x'^\alpha$ satisfy  (\ref{eq:s1}),(\ref{eq:s2}), we have $y_\alpha=y'_\alpha$ if and only if either $x^\alpha=x'^\alpha$ or $x^\alpha=-x'^\alpha,\, \alpha=1,...,5.$
\end{Lemma}
\begin{proof}The variables $y$ being quadratic in $x,$ it is clear that the 'if' part holds. Now suppose we have  $y_\alpha=y'_\alpha,\, \alpha=1,...,5.$ If $x^4=0,$ then $x'^4=0,$ therefore from (\ref{eq:s2}) we have that $x^5=\pm 1$ and $x'^5=\pm 1.$ It follows then from (\ref{eq:s1}) that $x^1=x^2=x^3=0,$ and the same for $x'.$ Therefore $x=(0,0,0,0,\pm 1)$ and  $x'=(0,0,0,0,\pm 1),$ thus $x'=\pm x.$ If $x^4\neq 0,$ then $x'^4/x^4=\pm 1$ and $y'_\alpha=(x'^4/x^4)y_\alpha.$\end{proof}
\subsection{Graphic representation}
To obtain a graphic representation we proceed as before and arrive, after renaming of the variables, at the following set of parametric equations
\begin{align}
x(\Psi,\Theta)&=\frac{2\cos^2\, \Psi}{2-\cos^2\,\Psi\cos\,\Theta}\\
y(\Psi,\Theta)&=\frac{2\cos^2\, \Psi\, \sin\, \Theta}{2-\cos^2\,\Psi\cos\,\Theta}\\
z(\Psi,\Theta)&=\frac{2\cos\, \Psi\, \sin\, \Theta}{2-\cos^2\,\Psi\cos\,\Theta}
\end{align}
The resulting surface has the shape of a simple elliptic supercyclide \textit{needle (horn) cyclide\,}  as in Fig. \ref{fig:fig2} - \cite[Fig. 6, p. 80]{schrott}, \cite[Fig. 5.7, p. 156]{cecil1}, or, in French, \textit{croissant simple\,} \cite{ferreol}. In $\BP^5$ the surface is, in fact, made of closed null geodesics, all intersecting at the point with homogeneous coordinates $({\bf 0},0,1,1)\sim  ({\bf 0},0,-1,-1).$ Each od these geodesics is uniquely determined by a point on the $2$-sphere $(\bn,1,0,0),$ $\bn^2=1.$ The geodesic is then given by the formula \be\gamma(\Psi)=[(\cos (\Psi)\bn,\cos{\Psi},\sin (\Psi),\sin(\Psi))],\quad \Psi\in [0,\pi]\label{eq:geo}\ee - cf. \cite[Eq. (16)]{ajci}. Taking another projection, switching the roles of  $x^1$ and $x^5,$ we arrive at a topologically equivalent, this time symmetric, representation - see Fig. \ref{fig:fig3}.
\begin{figure}[!ht]
\begin{center}
      \includegraphics[width=7cm, keepaspectratio=true]{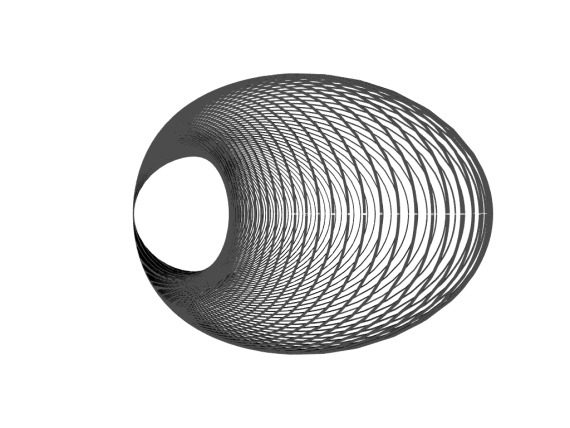}
\end{center}
  \caption{Pictorial representation of the simple conformal infinity  with one dimension skipped - needle cyclide, made of a one parameter family of null geodesics trapped at infinity.}
\label{fig:fig2}\end{figure}
\begin{figure}[!ht]
\begin{center}
      \includegraphics[width=7cm, keepaspectratio=true]{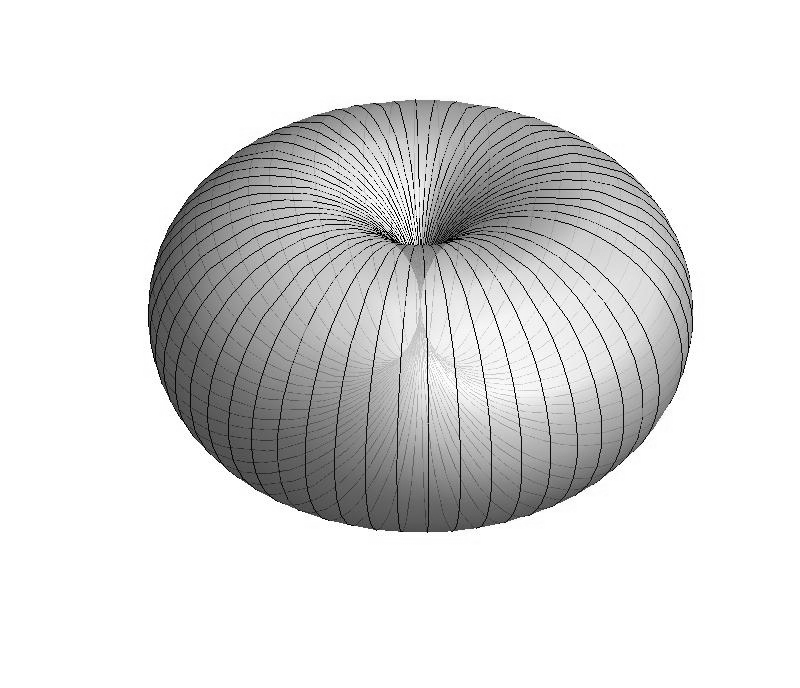}
\end{center}
  \caption{A symmetric representation of the simple conformal infinity, as a horned torus.}
\label{fig:fig3}\end{figure}
\section{Conformal infinity and Lie spheres\label{sec:ls}}
In 1872 Sophus Lie \cite{lie1} has formulated the geometry of oriented spheres in $\BR^3.$ It was further developed and generalized in the third volume of the monograph \cite{blaschke3} {\it `Differentialgeometrie der Kreise und Kugeln'\,}, published in 1929, by Wilhelm Blaschke. It's modern version is presented in {\it `Lie Sphere Geometry'\,} by Thomas E. Cecil \cite{cecil1}. Lie sphere geometry is concerned with the geometry of oriented spheres in $\BR^3$ (or, more generally in $\BR^n.$ An oriented sphere is a sphere with its radius vector pointing outwards (positive) or inwards (negative). A sphere of zero radius (no distinction between outwards and inwards) is just a point. An oriented  sphere of infinite radius is a plane - with its normal vector pointing in one or another direction. Added to points, oriented spheres, and oriented planes, is an exceptional point at infinity that makes $\BR^3$ into $S^3$ - it's one--point compactification. Formally, Lie sphere geometry is the study of the projective quadric $Q(x)=0$ and of the invariants of the action of $O(4,2)$ on this quadric.\\
 \noindent Blaschke \cite[p. 270]{blaschke3} noticed the relation of Lie sphere geometry to the Minkowski space of special relativity, but he did not elaborate much on this relation. The interpretation of relativistic space-time events in terms of Lie spheres can go as follows:  The radius $r$ can be interpreted as the radius of a spherical wave at time $t=r/c$, if the wave, propagating through space with the speed of light $c$, was emitted at $\bx,$ $|\bx|=r,$ at time $t=0.$ The image being that when the spherical wave reduces to a point, it turns itself inside--out, thus reversing its orientation.\\
 \noindent The correspondence between the constructs of Lie geometry in $\BR^{4,2}$ and geometrical objects in $\BR^3$ is given in the following table (adapted from \cite[p. 16]{cecil1}).\footnote{In \cite{ferap}  E. V. Ferapontov makes and interesting connection between Lie sphere geometry and twistor's formalism.}
\begin{table}[h!b!p!]
\caption{Correspondence between Lie spheres and points of the compactified Minkowski space. $[x]$ denotes the equivalence class modulo $\BR^*.$}
{\begin{tabular}{@{}cc@{}} 
\hline
\\[3pt]
\textbf{Euclidean}&\textbf{Lie}\\
points: $\bx\in\BR^3$&$[(\bx,0,\frac{1-\bx^2}{2},-\frac{1+\bx^2}{2},0,0)]$\\
\\[3pt]
$\infty$& $[(\mathbf{0},0,1,1)]$\\
\\[3pt]
spheres: center $\bx,$ signed radius $t$&$[(\bx,t,\frac{1-\bx^2+t^2}{2},-\frac{1+\bx^2-t^2}{2}$)]\\
\\[3pt]
planes: $\bx\cdot \bn=h,$ unit normal $\bn$&$(\bn,1,h,h)]$\\
\end{tabular}}
\end{table}
Conformal infinity of the Minkowski space consists of planes $\bx\cdot\bn=h,$ and of the point $\infty .$ According to Eq. (\ref{eq:geo}) all null geodesics trapped at infinity intersect at this special point, with $\Psi=\pi /2.$ For $\Psi\neq \pi /2$ the geodesic equation (\ref{eq:geo}) can be written as $\bx\cdot\bn=\tan\,\Psi.$ That means that a null geodesic trapped at infinity corresponds, in $\BR^3,$ to a family of parallel planes (plus $\infty$) - they represent light wave fronts - see Fig. \ref{fig:fig3}.
\begin{figure}[!ht]
\begin{center}
      \includegraphics[width=7cm, keepaspectratio=true]{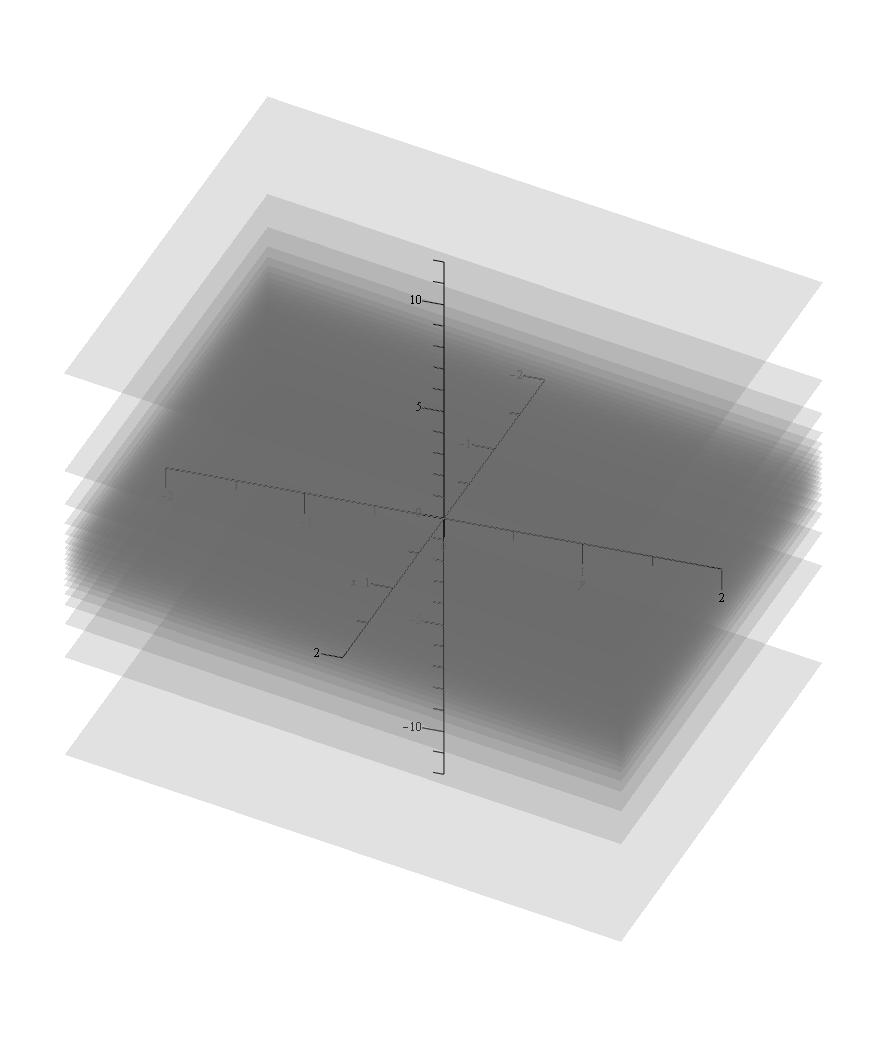}
\end{center}
  \caption{A family of plane fronts representing in $\BR^3$ a null geodesic $\bx\cdot\bn=\tan\,\Psi$ for $\bn=(1/\sqrt{2},0,1/\sqrt{2}),$ $\Psi=k*\pi/20,\,k=-9,...,9.$}
\label{fig:fig4}\end{figure}
The same family of fronts can be represented by the points on the null geodesic of the cyclide. In fact, there will be two geodesics, one for each of the two opposite orientation of planes - see Fig. \ref{fig:fig4}
\begin{figure}[!ht]
\begin{center}
      \includegraphics[width=7cm, keepaspectratio=true]{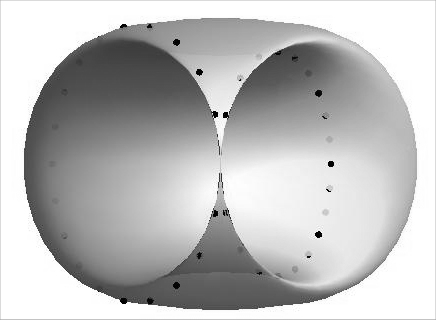}
\end{center}
  \caption{A family of points on the cyclide representing two null geodesics $\bx\cdot\bn=\tan\,\Psi$ for $\bn=\pm (1/\sqrt{2},0,1/\sqrt{2}),$ $\Psi=k*\pi/20,\,k=-9,...,9.$, this time viewed from a different perspective, so that the point $\infty$ is in front of the picture.}
\label{fig:fig5}\end{figure}
\subsection{Plane fronts}
Giving Minkowski's space conformal infinity the name of "the light cone at infinity" was unfortunate and misled even several expert authors of mathematical monographs. Is there a better picture? Using Eqs. (\ref{eq:s1},\ref{eq:s2}) we can parametrize conformal infinity by angle variables $\Psi\in [0,\pi],$ $\Theta,\Phi\in [0,2\pi]$ as follows:
\begin{align}
x^1&=\cos\, \Psi\,\sin\,\Theta\,\cos\,\Phi\\
x^2&=\cos\, \Psi\,\sin\,\Theta\,\sin\,\Phi\\
x^3&=\cos\, \Psi\,\cos\,\Theta\\
x^4&=\cos\,\Psi\\
x^5&=\sin\,\Psi ,
\end{align}
where we still need to identify $x$ with $-x.$ The whole information about the surface can be then expressed in terms of quadratic variables $y^i=x^ix^4,\, (i=1,2,3),$ and $y^4=x^5x^4.$ Thus conformal infinity is parametrized in $\BR^4$ as:
\begin{align}
y^1&=\cos^2\, \Psi\,\sin\,\Theta\,\cos\,\Phi\\
y^2&=\cos^2\, \Psi\,\sin\,\Theta\,\sin\,\Phi\\
y^3&=\cos^2\, \Psi\,\cos\,\Theta\\
y^4&=\cos\,\Psi\,\sin\, \Psi\\
\end{align}
By choosing stereographic projection with center at $(0,0,1,0)$ we can represent the family of null geodesics (parameter $\Psi$ varies along geodesics) in $\BR^3,$ missing only one point, as follows:
\begin{align}
x&=\cos^2\, \Psi\,\sin\,\Theta\,\cos\,\Phi/(1-\cos^2\, \Psi\,\cos\,\Theta)\\
y&=\cos^2\, \Psi\,\sin\,\Theta\,\sin\,\Phi/(1-\cos^2\, \Psi\,\cos\,\Theta)\\
z&=\cos\,\Psi\,\sin\, \Psi/(1-\cos^2\, \Psi\,\cos\,\Theta)
\end{align}
with the following graphic representation:
\begin{figure}[!ht]
\begin{center}
      \includegraphics[width=7cm, keepaspectratio=true]{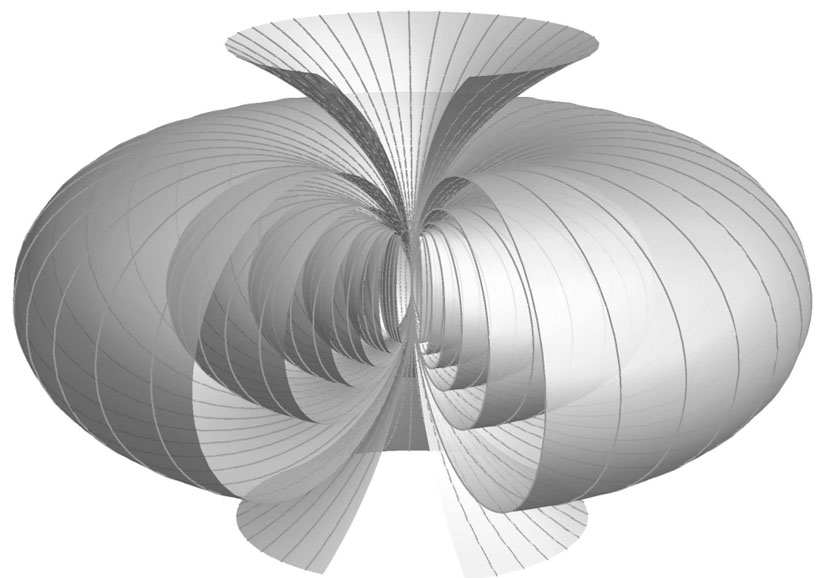}
\end{center}
  \caption{Conformal infinity represented in $\BR^3$}
\label{fig:fig6}\end{figure}
The figure resembles Clifford--Hopf fibration (cf. e.g. \cite[Fig. 33.15]{penroserd}, but is essentially different. The circles here are not the Villarceau circles (or `Clifford parallels') and the tori are limit tori with one common point - the point $\infty.$
\section{Compactified Minkowski space and it's conformal infinity in $1+1$ space--time dimensions}
In $1+1$ space--time dimensions, with coordinates $(x,t)$ the compactified Minkowski space is described, in $\BR^{2,2}$
with coordinates $(X,T,V,W),$ by equations (cf. Eq. (\ref{eq:ss}) \begin{align} X^2+V^2&=1,\nonumber\\
T^2+W^2&=1.\label{eq:xtvw}\end{align}
We should then identify $(X,T,V,W)$ with $(-X,-T,-V,-W).$ It is convenient to introduce complex variables $z_1=X+iV,$ $z_2=T+iW,$ with $|z_1|=|z_2|=1.$ The necessity of identification may seem, at first sight, to complicate the picturing of the surface. What we have is the Clifford torus quotiented by $Z_2$ action $f:(z_1,z_2)\mapsto (-z_1,-z_2).$ However, the following lemma is easy to prove.
\begin{Lemma}
The map $(z_1,z_2)\mapsto (z_1z_2,z_1\bar{z}_2)$ is a surjection from the Clifford torus onto itself. The counterimage of each point consists of exactly two points $(z_1,z_2)$ and $(-z_1,-z_2).$
\end{Lemma}
It follows that $M^c$ is, in our case, nothing else but the Clifford torus. We can represent it now in $\BR^3$ using stereographic projection, but is  more instructive to embed first the Minkowski space $M$ in $M^c.$ To this end we first embed $M$ into the isotropic cone of $\BR^{2,2},$ in the standard way (cf. Eq. (\ref{eq:s1})) $(x,t)\mapsto (x,t,v,w)=(x,t,(1-x^2+t^2)/2,-(1+x^2-t^2)/2).$ We then have, automatically, $x^2+v^2=t^2+w^2>0.$ In order to have (\ref{eq:xtvw}) satisfied we introduce normalized variables $(X,T,V,W)=(x,t,v,w)/\sqrt{t^2+w^2},$ then $z_1=X+iV,z_2=T+iW,$ and plot $(\Re (z_1z_2),\Im (z_1z_2),\Re (z_1\bar{z}_2),\Im (z_1\bar{z}_2))$ using stereographic projection from four to three dimensions with center at $(2,0,0,0).$ In Fig. \ref{fig:fig7} we plot this way the part of Minkowski space corresponding to the  rectangle $|x|\leq 20,|t|\leq 15.$
\begin{figure}[!ht]
\begin{center}
      \includegraphics[width=5cm,keepaspectratio=true]{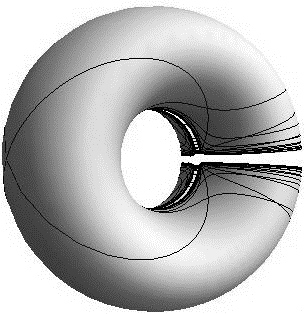}
\end{center}
  \caption{$1+1$ dimensional Minkowski space on the Clifford torus representing $M^c.$}
\label{fig:fig7}\end{figure}
The remaining part contains conformal infinity which, in this case, is represented by two circles with one common point: $\infty.$

In Segal's model \cite[Ch. III.5]{segal}, cf. also \cite{werth}, an important role is being played by the temporal evolution emerging from the action of the circle group on the $S^1\times S^3.$ For our $1+1$ dimensional model this action corresponds to the multiplication by $z_1.$ The corresponding orbits on $M^c$ are then Villarceau circles - see Fig. \ref{fig:fig8}:

\begin{figure}[!ht]
\begin{center}
      \includegraphics[width=5cm,keepaspectratio=true]{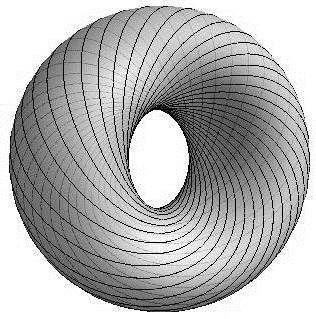}
\end{center}
  \caption{Trajectories of unispace Segal's dynamics on Clifford's torus representing $M^c.$}
\label{fig:fig8}\end{figure}
\section{Acknowledgments}
The author thanks Pierre Angl\`{e}s, Robert Coquereaux, Lionel Garnier, Marek Golasi\'{n}ski and Alexander Levichev for helpful comments, and also acknowledges support by Quantum Future Group.


\newpage
\begin{thebibliography}{99}
\bibitem{penrose1} Roger Penrose, The Light Cone at Infinity, in {\it Relativistic Theories of Gravitation\,}, ed. L.~Infeld, Pergamon Press, Oxford, 1964, pp.~369--373,
\bibitem{akivis} Maks A. Akivis, Vladislav V. Goldberg, {\it Conformal Differential Geometry and its Generalizations\,}, A Wiley Interscience Publications, New York, 1996
\bibitem{tod} S.~A.~Huggett and K.~P.~Tod, {\it An Introduction to Twistor Theory\,}, Cambridge University Press, 1994
\bibitem{ajci} Arkadiusz Jadczyk, {\it On Conformal Infinity and Compactifications of the Minkowski Space\,}, Advances in Applied Clifford Algebras, DOI: 10.1007/s00006-011-0285-5, 2011, \url{http://arxiv.org/abs/1008.4703}
\bibitem{aj9a} Arkadiusz Jadczyk, {\it Conformally Compactified Minkowski Space: Myths and Facts}, \url{http://arxiv.org/abs/1105.3948}
\bibitem{angles} Pierre Angl\`es, {\it Conformal Groups in Geometry and Spin
Structures}, Birkhauser, Progress in Mathematical Physics, Vol. 50, 2008
\bibitem{lie1}Sophus Lie, {\it \"{U}ber Komplexe, inbesondere Linien- und Kugelkomplexe, mit Anwendung auf
der Theorie der partieller Differentialgleichungen\,}, Math. Ann., \textbf{5}  (1872), pp.~145�-208,~209�-256 (Ges. Abh. 2, 1�121)
\bibitem{blaschke3}Wilhelm Blaschke, \textit{ Vorlesungen \"{u}ber Differentialgeometrie und geometrische Grundlagen von Einsteins Relativit\"{a}tstheorie\,}, Vol. \textbf{3}, Springer-Verlag, Berlin, 1929
\bibitem{cecil1}Thomas E. Cecil, \textit{Lie Sphere Geometry}, Second edition, Springer--Verlag, New York, 2000
\bibitem{uhlmann63} Armin Uhlmann, {\it The Closure of Minkowski Space\,}, Acta Physica Polonica, {\bf XXIV}, Fasc. 2(8), (1963), pp.~295--296,
\bibitem{bateman1}H. Bateman, \textit{The Conformal Transformations of a Space of Four Dimensions and their Applications to Geometrical Optics\,}, Proc. London. Math. Soc. (ser. 2), \textbf{7}, 70--98, (1909)
\bibitem{cunningham1}E. Cunningham, \textit{The Principle of Relativity in Electrodynamics and Extension Thereof\,}, Proc. London. Math. Soc. (ser. 2), \textbf{8}, 77--98, (1910)
\bibitem{bateman2}H. Bateman, \textit{The Transformation of the Electrodynamical Equations\,}, Proc. London. Math. Soc. (ser. 2), \textbf{8}, 223--264, (1910)
\bibitem{cj} R. Coquereaux and A. Jadczyk, {\it Conformal Theories, Curved Phase Spaces, Relativistic Wavelets and the Geometry
of Complex Domains\,}, Rev.~Math.~Phys., {\bf 2}, No 1 (1990), pp.~1--44
\bibitem{wk} W.~Kopczy\'{n}ski and L.~S. Woronowicz, {\it A geometrical approach to the twistor formalism\,}, Rep. Math. Phys., {\bf 2} (1971), pp.~35--51,
\bibitem{penrin} R. Penrose and W. Rindler, {\it Spinors and Space-Time, Vol. 2 -- Spinor and Twistor Methods in Space-Time Geometry}, Cambridge University Press, Cambridge, England, 1984
\bibitem{garnier}L. Garnier, {\it Math\'ematiques pour la mod\'elisation g\'eom\'etrique, la
repr\'esentation 3D et la synth\`{e}se d'images}, Ellipses, Paris, 2007
\bibitem{schrott}Michael Schrott, Boris Odehnal, {\it Ortho-Circles or Dupin Cyclides\,}, Journal of Geometry and Graphics, \textbf{1}, (2006), pp.~73--98
\bibitem{ferreol}Robert Ferr\'{e}ol, \textit{CYCLIDE DE DUPIN, Dupin's Cyclide, dupinsche Zyklide\,}, \url{http://www.mathcurve.com/surfaces/cycliddedupin/cyclidededupin.shtml}
\bibitem{pratt}M. J. Pratt, {\it Cyclides  in  computer  aided geometric  design\,}, Computer  Aided  Geometric  Design,  \textbf{7} (1990), pp.~221-242
\bibitem{artmann}Benno Artmann, {\it Pictures of the Projective Plane\,}, The Montana Mathematics Enthusiast, FESTSCHRIFT IN HONOR OF G\"{U}NTER T\"{O}RNER'S
60th BIRTHDAY, TMME Monograph 3 (2007), pp.~3-16, \url{http://www.math.umt.edu/tmme/Monograph3/Artmann_Monograph3_pp.3_16.pdf}
\bibitem{hilbert}D. Hilbert and S. Cohn-Vossen, {\it Geometry and Imagination\,}, Chelsea Publishing Company, New York, 1990
\bibitem{ferap}E.~V.~Ferapontov, {\it The analogue of Wilczynski's projective frame in Lie sphere geometry: Lie-applicable surfaces and commuting Schr�dinger operators with magnetic fields\,}, \url{http://arxiv.org/abs/math/0104034}
\bibitem{penroserd}Roger Penrose, {\it The Road to Reality\,}, Jonathan Cape, London, 2004
\bibitem{segal}Irving Ezra Segal, {\it Mathematical Cosmology and Extragalactic Astronomy\,}, Academic Press, New York, 1976
\bibitem{werth}J.~E. Werth, {\it Conformal group actions on Segal's cosmology\,}, Rep. Math. Phys., {\bf 23} (1986), pp.~257--268
\end{thebibliography}
\end{document}